\documentclass{article}
\usepackage[a4paper,left=3cm,right=3cm,top=3cm,bottom=3cm]{geometry}

\usepackage{amsmath,amsthm}
\usepackage{amssymb}
\usepackage{times}
\usepackage{bm}
\usepackage{natbib}
\usepackage{graphicx}
\usepackage[plain,noend]{algorithm2e}

\def\mA{\mathcal{A}}
\def\mB{\mathcal{B}}

\def\Prob{\mathrm{pr}}

\def\R{\mathbb{R}}

\def\Ind{\mathbb{I}}
\def\Real{\mathbb{R}}
\DeclareMathOperator*{\argmax}{arg\,max}

\def\Tr{\mathsf{T}}
\def\Cr{\mathsf{C}}

\def\mA{\mathcal{A}}

\newtheorem{lemma}{Lemma}
\newtheorem{corollary}{Corollary}

\usepackage{authblk}
\author[*]{Patrick Rubin-Delanchy}
\author[**]{Nicholas A Heard}
\affil[*]{Heilbronn Institute for Mathematical Research, University of Bristol, U.K.}
\affil[**]{Department of Mathematics, Imperial College London, U.K.}
\date{}
\title{A test for dependence between two point processes on the real line}

\begin{document}
\maketitle
\begin{abstract}
Many scientific questions rely on determining whether two sequences of event times are associated. This article introduces a likelihood ratio test which can be parameterised in several ways to detect different forms of dependence. A common finite-sample distribution is derived, and shown to be asymptotically related to a weighted Kolmogorov-Smirnov test. Analysis leading to these results also motivates a more general tool for diagnosing dependence. The methodology is demonstrated on data generated on an email network, showing evidence of information flow using only timing information. Implementation code is available in the R package `mppa'.
\end{abstract}


\section{Introduction}
Testing for dependence between two point processes is a long-standing statistical problem. When the two processes are on the real line, usually representing time, the points are usually interpreted as event times or, simply, events. Scientific questions then often revolve around identifying triggering behaviour (the occurrence of an event in $A$ temporarily increases the rate of events in $B$), correlation (the rate of events in $B$ is locally increased around events in $A$), inhibition, anti-correlation and so on.

Statistical methods to detect such effects have received decades of attention in the field of neurophysiology. The activity of a neuron is often recorded as a sequence of `spike' times, called a neuronal spike train, which is often treated as a realisation of a point process on the real line. Comparing trains that are simultaneously generated by different neurons can shed light on how they are connected and, more generally, how information is processed in the nervous system. The literature in this field is relatively mature, for example a very highly cited paper by \citet{perkel67} proposed to test for interaction on the basis of histograms of the times from $A$ to $B$ events. Since then, a number of model-based approaches were developed for this problem, notably in a series of papers by Brillinger \citep{brillinger1976identification, brillinger1979empirical, brillinger1988some, brillinger1988maximum, brillinger1992nerve}.

In a more general context \citet{ripley76, ripley77} introduced the so-called $K$-function to measure second-order dependence between point processes defined on a topological space. From this work a number of articles followed, typically motivated by ecological or biological applications, adapting ideas to two- or three-dimensional settings \citep{lotwick82, berman86}. \citet{doss89}, again motivated by a neurophysiological application, provided an interesting asymptotic analysis of the estimated $K$-function for point processes on the real line. 

The problem is now critical in the analysis of network data, for example traffic generated on a computer network, messages and other connections made on social networks, mobile communications, email networks, the web, collaboration networks (e.g. in academia, music or film) and more. Such data can often be represented as a graph with point processes (e.g. communication times) occurring on every edge (e.g. a pair of computers). In being able to diagnose dependence between events generated by edges or nodes, there is great potential to better understand information flow, discover new correlations and develop more accurate network models. Recent approaches include \citet{blundell12}, where email reciprocation is modelled using the mutually exciting point process models developed by \citet{hawkes71bis, hawkes71}, and \citet{perry13},  developing a framework for modelling point process interaction networks based on a version of Cox's proportional intensity model. 

This article seems to be the first to propose an exact (generalised) likelihood ratio test for association between two one-dimensional point processes. The user supplies a model for $B$, which captures its statistical behaviour under the null hypothesis, for example encapsulating any seasonality, changepoints or drift. Then the procedure tests for a multiplicative effect on the intensity of $B$ within a certain interval following or surrounding every event in $A$. Various different forms of dependence can be identified by simple modifications of the procedure. 

The contributions of this paper are very practical, for example the p-value is exact in finite samples and results are valid under non-homogeneous conditions. Asymptotically, the procedure provides the uniformly most powerful test and is related to a weighted Kolmogorov-Smirnov (K-S) test. A key insight is noticing a duality with finding a change-point in a homogeneous Poisson process (Lemma \ref{lem:uniform}), after which mathematical considerations are greatly simplified. A by-product of the lemma is a new diagnostic tool for analysing dependence between point processes.

\section{Testing for triggering behaviour}
\subsection{A fixed range interaction model}
\label{sec:triggering}
Let $A$ and $B$ be two simple point processes on the real line observed simultaneously from the first event time of $A$ up to an observation end time $L$. Neither process is explosive, so that the observed event times of $A$ form a finite set $\mA=\{a_1< \ldots < a_m\}$, $m\geq 1$, and the event times of $B$ form a finite set $\mB=\{b_1< \ldots < b_n\}$ where $b_1 \geq a_1$ and $n \geq 0$. Without loss of generality assume $a_1 = 0$.  We make no further assumptions on $A$, and treat its event times as given when modelling $B$. 

In this section we consider the problem of testing whether events in $A$ cause an increase in the intensity of $B$ (triggering behaviour) under the assumption that $B$ conditional on $A$ is a non-homogeneous Poisson process. Tests for other forms of dependence and relaxations of the Poisson assumption are considered in Section \ref{sec:extensions}. Given $A$, the process $B$ is assumed to have a deterministic, bounded, Lebesgue-measurable intensity function  $\lambda_B(t)$, 
\begin{equation}
\lambda_{B}(t) = \begin{cases} \lambda_1 r(t) & t-a(t) \leq \tau,\\
\lambda_2 r(t) & t-a(t) > \tau,
\end{cases} \quad t \in [0,L),
\label{eq:model}
\end{equation}
where $a(t)$ is the most recent event in $\mA$ occurring at or before $t$,  $\tau>0$, $\lambda_1 \geq \lambda_2 \geq 0$ are unknown parameters, and $r$ is a known bounded non-negative Lebesgue-measurable function satisfying $\int_0^L r(v) d v =1$.

Model \eqref{eq:model} leads to a test with a very straightforward interpretation: is the relative proportion of events within time $\tau$ of an event in $\mA$ higher than can be explained by $r$ alone? This is formalised as the following hypothesis test:
\begin{equation}
H_0: \lambda_1 = \lambda_2 \quad \text{versus} \quad H_1: \lambda_1 > \lambda_2.\label{eq:hyptest}\end{equation}
A natural test statistic is the generalised likelihood ratio
\begin{equation}
\frac{\sup\left\{\ell(B; \tau,\lambda_1,\lambda_2) \: : \:\tau>0, \:\lambda_1>\lambda_2 \geq 0 \right\}}{\sup\left\{\ell(B; \tau,\lambda_1,\lambda_2) \: : \:\tau>0, \:\lambda_1=\lambda_2 \geq 0 \right\}}, \label{eq:test}
\end{equation}
where $\ell$ is the likelihood of $B$ under model \eqref{eq:model}. Under $H_0$, $\tau$ has no real importance and, accordingly, $\ell(B; \tau,\lambda_1,\lambda_2)$ is functionally independent of $\tau$ when $\lambda_1 = \lambda_2$.

A common choice for $r$ will be the constant $1/L$, in which case $B$ is a homogeneous Poisson process under $H_0$. If the application makes this assumption unrealistic, a more informed choice of $r$ will not only bring the null behaviour of the test closer to its nominal distribution, derived later, but may also lead to a gain in power under the alternative, for example if $B$ appears to respond to $A$ despite being relatively inactive under $H_0$. In Section~\ref{sec:enron} a Bayesian estimate of $r$ is used. A more formal treatment of the case where $r$ is unknown is outlined in the Discussion.

\subsection{Computation of the test statistic}
\label{sec:explicit}
Let $\rho(X) = \int_X r(v) \:d v$, where $X$ is a Borel set. $\rho$ can be seen as an artificial measurement of time that compensates for the varying intensity of $B$ under $H_0$. In fact, $\rho$ is a probability measure on $[0,L)$. The likelihood of model (\ref{eq:model}) is \citep[p. 232]{daley03}
\begin{equation}
\ell(B;\tau,\lambda_1,\lambda_2) \propto \lambda_1^{K(\tau)}\exp[-\lambda_1 \rho\{\Tr(\tau)\}] \times \lambda_2^{n-K(\tau)} \exp[-\lambda_2 (1-\rho\{\Tr(\tau)\})], \label{eq:likelihood}
\end{equation}
where $\Tr(\tau) =  \{t: t-a(t) \leq \tau\}$ is the union of all triggered intervals and $K(\tau) = \#\{b_i \in \Tr(\tau)\}$ is the number of triggered events.

Let $u_1 \leq \cdots \leq u_n$ be the order statistics of $\rho\{\Tr(b_i - a(b_i))\}, i = 1, \dots, n$. The variable $u_k$ can be interpreted as the effective proportion of triggered time if $\tau$ is equal to the $k$th smallest response time. The maximum likelihood parameters $\hat{\tau}, \hat{\lambda}_1, \hat{\lambda}_2$ are found in Algorithm \ref{alg:lik}, for $n \geq 1$.

\begin{algorithm}[!h]
\caption{Computation of $T$ given $u_1, \dots, u_n$, for $n \geq 1$}
\label{alg:lik}
For $k = 1, \dots, n$, let
\[\ell_k=\left(\frac{k/n}{u_{k}}\right)^{k/n} \left(\frac{(n-k)/n}{1-u_{k}}\right)^{(n-k)/n}.\] 
Then, let
\[ \hat{k} = \argmax_{k = 1, \dots, n}\{\ell_k: u_k \leq k/n\}, \quad
\hat{\tau} = b_{\hat{k}}-a(b_{\hat{k}}), \quad \hat{\lambda}_1 = \hat{k}/u_{\hat{k}}, \quad \hat{\lambda}_2 = (n-\hat{k})/(1-u_{\hat{k}}). \]
Return $T=\ell_{\hat{k}}$.
\end{algorithm}
In Algorithm \ref{alg:lik}, because $k=n$ satisfies $u_k \leq k/n$, the estimate $\hat{k}$ is always defined. The maximum of \eqref{eq:likelihood} is a monotonic function of $\ell_{\hat{k}}$. This can be shown by a straightforward argument, given in the supplementary material. Since the number of events in $B$ can be equally well explained under the null as under the alternative, it is natural to condition on the value of $n$. The denominator of \eqref{eq:test} is then constant, therefore any monotonic function of the numerator can be used a test statistic, and we use $T=\ell_{\hat{k}}$.

\subsection{A simple reformulation}
The following lemma establishes a duality between the hypothesis test set out in \eqref{eq:hyptest} and the problem of testing for a Poisson process change-point, and is the key observation of this article. The proof is given in the appendix.
\begin{lemma}
\label{lem:uniform}
$u_1, \dots, u_n$ are the event times of a Poisson process $U(x)$ on $[0,1)$ with a change-point in its intensity,
\begin{equation}
\lambda(x) = \begin{cases} \lambda_1 & x \leq \rho\{\Tr(\tau)\}, \label{eq:equivalence}\\
\lambda_2 & x > \rho\{\Tr(\tau)\}, \end{cases}
\end{equation}
for $x \in [0,1)$.
\end{lemma}
Lemma \ref{lem:uniform} and its proof provide a number of insights into the testing problem. First, the statistic $T$ is also the generalised likelihood ratio test for model \eqref{eq:equivalence} against a homogeneous Poisson null hypothesis,
\begin{equation*}
\frac{\sup\left\{\ell(U; \tau,\lambda_1,\lambda_2) \: : \:\tau>0, \:\lambda_1>\lambda_2 \geq 0 \right\}}{\sup\left\{\ell(U; \tau,\lambda_1,\lambda_2) \: : \:\tau>0, \:\lambda_1=\lambda_2 \geq 0 \right\}}. 
\end{equation*}
Second, conditional on $n$, the variables $u_1, \ldots, u_n$ are ordered uniform random variables under the null hypothesis, whereas under the alternative they should be, loosely speaking, more concentrated towards $0$. Hence the $u_i$ provide a more general tool for diagnosing dependence. For example they can be used in a goodness-of-fit test against uniformity, e.g. Fisher's method \citep{Fisher48}, or in a more visual way, e.g. a plot of the empirical cumulative distribution function of $u_i$ compared to $y=x$. Finally, Lemma \ref{lem:uniform} makes it relatively straightforward to determine conditions for consistency and the asymptotic optimality of the test, given below and proven in the appendix. To give a more compact statement, we have temporarily set $a_{m+1}=L$ below.
\begin{corollary}\label{cor:consistence}
Suppose that $0 < \tau < \max(a_{i+1}-a_i: i = 1, \ldots, m)$ and $r$ is positive in the neighbourhood of a change, i.e., an open interval containing a point $t$ satisfying $t-a(t) = \tau$. In the asymptotic regime $\lambda_1,\lambda_2 \rightarrow \infty$ with $\lambda_1/\lambda_2=c>1$, the estimate $\hat \tau$ is consistent and $T$ becomes a monotonic function of the true likelihood ratio conditional on $n$.
\end{corollary}
\subsection{Finite sample p-value}\label{sec:exact}
The p-value is one if $n=0$. If $n\geq 1$ the proposed test has a p-value $p = \Prob(T \geq t \mid n) = 1-F_n(t)$, where $t$ is the observed test statistic, $T$ is a replicate of $t$ under the null hypothesis, and $F_n$ is the cumulative distribution function of $T$ under the null hypothesis conditional on $n$. For $n\geq 1$, the p-value can be computed explicitly using
\begin{equation*}
F_n(t) = \Prob[u_1 \geq o_1, \dots, u_n \geq o_n],
\end{equation*}
where $o_i$ is the solution for $x \in (0, i/n]$ of 
\begin{equation}
t=\left(\frac{i/n}{x}\right)^{i/n} \left(\frac{(n-i)/n}{1-x}\right)^{(n-i)/n}, \label{eq:o}
\end{equation}
which is obtained numerically. Various recursive formulas exist for computing the joint survival probability of $n$ ordered uniform variables, although many are unsuitable for computation because they involve differences of very large numbers (so-called catastrophic cancellation). A safe option is the $O(n^2)$ formula in \citet{noe68}, as corrected in \citet{noe72}. This recursion is implemented in the R package corresponding to this article, `mppa'. \citet{worsley88} proposed a similar idea in the context of testing for a changepoint in the hazard rate of survival times.
\section{Extensions}\label{sec:extensions}
\subsection{Time-limited $\tau$}\label{sec:time_limit}
It may be desirable to limit $\tau$ to a maximum range, $\tau_{\max}$ say. This avoids wasting power on testing for long-term dependence if $\tau$ is expected to be small. In this case the test statistic is computed as follows. Let $u_{\max}=\rho\{\Tr(\tau_{\max})\}$. Modify Algorithm~\ref{alg:lik} so that, if no $u_k\leq u_{\max}$, the returned value is $1$. Otherwise, replace  $\hat{k}$ with $\hat{k} = \argmax_{k = 1, \dots, n}\{\ell_k: u_k \leq \min(k/n, u_{\max})\}$.

Set the p-value of this test to be $p=1$ if $T=1$ (to be conservative). Otherwise compute $o_1, \dots, o_n$ as in \eqref{eq:o}, and calculate 
\[p=1-\Prob[u_1 \geq \min(o_1, u_{\max}), \dots, u_n \geq \min(o_n, u_{\max})].\]

\subsection{Testing for correlation}
The test can be modified to identify correlation, defined here to be an increased rate of events in $B$ surrounding events in $A$. We do not find it problematic that a test for correlation analysing $A$ conditional on $B$, instead of $B$ conditional $A$, could give a different result, because we see the two approaches as answering slightly different questions.

Relax the constraint $a_1=0$. Let $\tilde a(t)$ be the closest event to $t$ in $\mA$, which can now occur before or after $t$ and then replace $t-a(t) \leq \tau$ by $|t-\tilde a(t)| \leq \tau$ in \eqref{eq:model}. A generalised likelihood ratio test of $H_0: \lambda_1 = \lambda_2$ versus $H_1: \lambda_1>\lambda_2$ is obtained as follows. Let $\Cr(\tau) = \{t: |t-\tilde a(t)| \leq \tau\}$ and let $v_1 \leq \cdots \leq v_n$ be the order statistics of $\rho\{\Cr(b_i - \tilde a(b_i))\}$, for $i = 1, \dots, n$. Compute $T$ by inputting $v_1, \dots, v_n$  to Algorithm~\ref{alg:lik}. By a similar argument to the proof of Lemma \ref{lem:uniform}, we find that $v_1,\ldots, v_n$ are the event times of a point process following model \eqref{eq:equivalence}. Thus $T$ has distribution $F_n$ under $H_0$, conditional on $n$.

Further examples of how the procedure can be modified to detect other forms of dependence are given in the supplementary material.

\subsection{Independence conditional on $n$}
Relaxing the non-homogeneous Poisson assumption, suppose that $B$ can be generated by drawing $n$ from some distribution and then placing the $n$ event times independently according to some probability measure over $[0,L)$ with density
\begin{equation}
\label{eq:density_model}
d_{B}(t) \propto \begin{cases} \lambda_1 r(t) & t-a(t) \leq \tau,\\
\lambda_2 r(t) & t-a(t) > \tau,
\end{cases}
\end{equation}
where $r$ is as before a bounded non-negative Lebesgue measurable function satisfying $\int_0^L r(v) d v =1$, thereby defining $\rho$. If $n$ has a Poisson distribution this model reduces to \eqref{eq:model}. The hypothesis test $H_0: \lambda_1 = \lambda_2$ versus $H_1: \lambda_1 > \lambda_2$ can be evaluated through a similar generalised likelihood ratio test: we first compute $u_1\leq \dots \leq u_n$ as the order statistics of $\rho\{\Tr(b_i - a(b_i))\}, i = 1, \dots, n$ and then $T$ using Algorithm~\ref{alg:lik}. By  straightforward modifications to the proof of Lemma \ref{lem:uniform}, we find:
\begin{lemma}
\label{lem:density}
Under model \eqref{eq:density_model}, conditional on $n$, $\{u_i\}$ is a set of independent and identically distributed random variables with support on $[0,1)$ and density
\begin{equation*}
d(x) \propto \begin{cases} \lambda_1 & x \leq \rho\{\Tr(\tau)\},\\
\lambda_2 & x > \rho\{\Tr(\tau)\}. \end{cases} 
\end{equation*}
\end{lemma}
\noindent From this we establish that $T$ given $n$ also has null distribution $F_n$ under model \eqref{eq:density_model}.

\subsection{Random time transformation}\label{sec:tt}
A much wider class of point processes can be conceived by allowing the intensity of $B$ to be dependent on past information. Under some regularity conditions, $B$ has a continuous compensator $\Lambda(t)$ and a conditional intensity defined via $\Lambda(t) = \int_0^{t} \lambda(x) dx$ \cite[p.358, p.367, p.390]{daley07}. Lemma \ref{lem:uniform} and the results that follow do not continue to hold if $r$ is replaced by $\lambda$. This is illustrated in a simple example, drawn out in the supplementary material, where $B$ is a point process with just one point uniformly distributed on $[0,L)$, and $\mA = \{0, L/2\}$. On the other hand, probabilistic structure due to the conditional intensity of $B$ can be removed using the random time transformation theorem \citep[p.421]{daley07}: if a point process $X$ is non-terminating (there are infinite events as $t \rightarrow \infty$) with continuous compensator $C$ then the process $X\{C^{-1}(t)\}, t \in [0,\infty)$ is a homogeneous Poisson process with unit rate, where $F^{-1}(y) = \inf \{x: F(x)\geq y\}$ for a non-decreasing function $F$.

Thus if $\Lambda$ is known and continuous under the null hypothesis then $\tilde B(t) = B(\Lambda^{-1}(t)), t \in [0, \Lambda(L))$ is a stopped unit-rate homogeneous Poisson process. As an alternative hypothesis where events in $A$ trigger the intensity of $B$, we might propose the following model for the intensity of $\tilde B$,
\begin{equation*}
\lambda_{\tilde B}(t) = \begin{cases} \lambda_1/\Lambda(L) & t-\tilde a(t) \leq \tau,\\
\lambda_2/\Lambda(L) & t - \tilde a(t) > \tau, 
\end{cases}\quad t \in [0,\Lambda(L)),
\end{equation*}
where $\tilde a(t)$ is the most recent event in $\tilde \mA$ occurring at or before $t$, and $\tilde \mA$ is the set of transformed event times of $\tilde A(t) = A(\Lambda^{-1}(t))$.
The dependence of the stopping-time $\Lambda(L)$ on $B$ under $H_0$ makes inference more complicated. In particular Lemma~\ref{lem:uniform} and the p-value computed in Section~\ref{sec:exact} no longer hold exactly, but may be sufficiently close approximations for practical use. Progress is possible if we allow $L$ to be random, but this seems contrived.

\section{Asymptotic distribution}\label{sec:asymp}
\begin{figure}[htp]
  \centering
  \includegraphics[width=7cm]{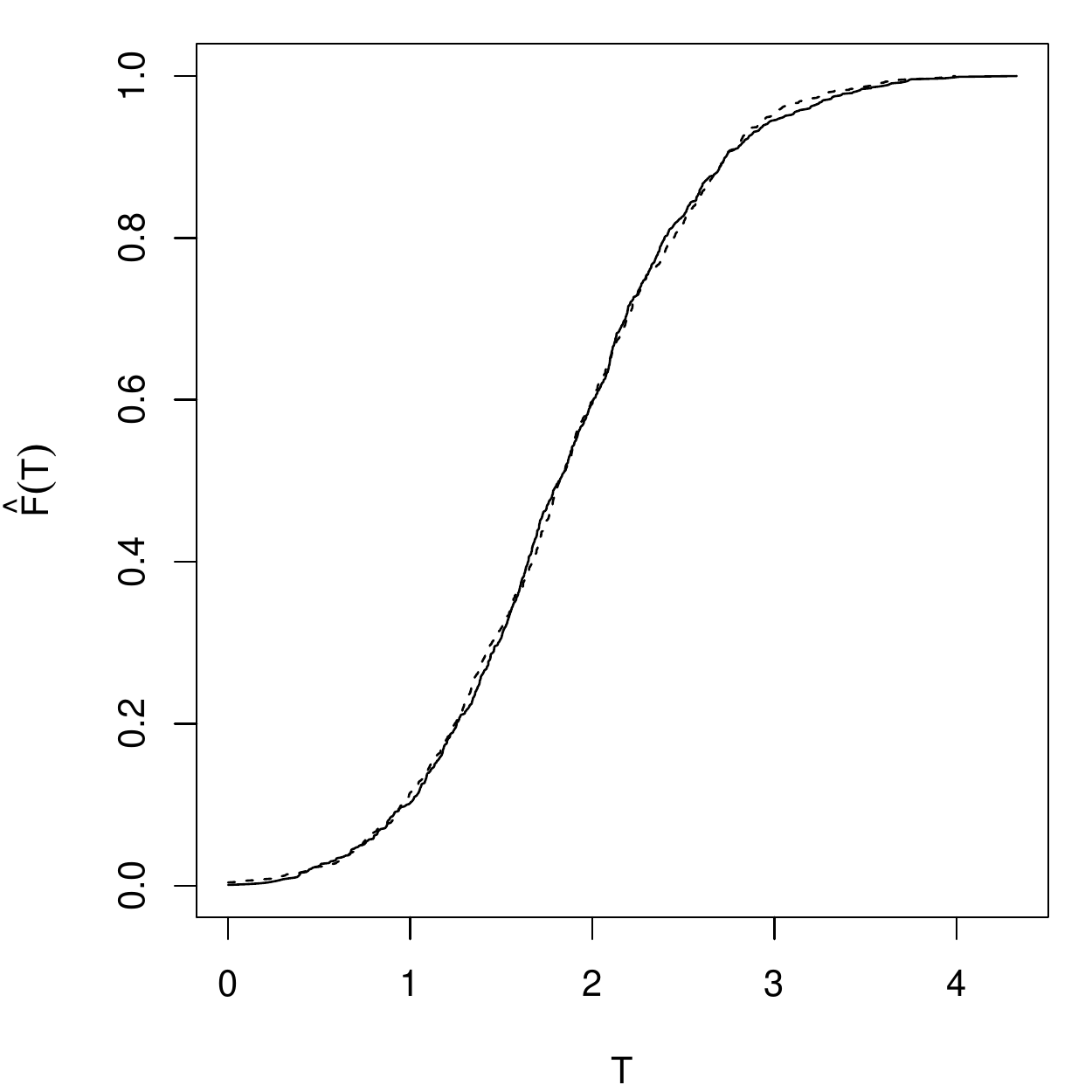}
  \caption{Empirical cumulative distribution function based on 1000 replicates of $(2n)^{1/2}[T_{[\gamma_1, \gamma_2]}-1]^{1/2}$ (solid line) and the weighted upper K-S test (dashed line), with $n=1000$ and $[\gamma_1, \gamma_2]=[.01, .99]$.}
  \label{fig:asymp}
\end{figure}

In this section we demonstrate an asymptotic connection between our test and a weighted Kolmogorov-Smirnov (K-S) test. The weight function needed is in fact one of the most frequently used for K-S tests, for example in \cite{anderson52} or \citet{chicheportiche12}.

Let $y_1\leq \cdots \leq y_n$ be the order statistics of independent replicates of an absolutely continuous random variable $Y$ with support $\mathbb{Y} \subseteq \Real$ and distribution function $F$, estimated as
\[\hat{F}(y) = \frac{1}{n} \sum_{i=1}^n \Ind[y_i \leq y],\]
where $\Ind$ is the indicator function. The generalised K-S test \citep{anderson52} is commonly used to test whether $\hat F$ is consistent with $F$,
\[G=\sup\left\{n^{1/2} |\hat{F}(y)-F(y)|[\phi\{F(y)\}]^{1/2}: y \in \mathbb{Y}\right\},\]
for some weight function $\phi(x)\geq 0, x \in [0,1]$. A one-sided, interval-restricted version of the above is
\begin{equation*}
G^+_{[\gamma_1, \gamma_2]}=\max\left\{n^{1/2} [\hat{F}(y_i)-F(y_i)] [\phi\{F(y_i)\}]^{1/2}: F(y_i) \leq \hat{F}(y_i); \gamma_1 \leq F(y_i) \leq \gamma_2\right\}, 
\end{equation*}
for $0<\gamma_1 < \gamma_2 < 1$, defining $G^+_{[\gamma_1, \gamma_2]}=0$ if the above set is empty. 

Analogously, let
\[T_{[\gamma_1, \gamma_2]} = \max\{\ell_i: u_i \leq i/n; \gamma_1 \leq u_i \leq \gamma_2\},\]
now setting $T_{[\gamma_1, \gamma_2]}=1$ if the above set is empty. We next show, by some slightly heuristic arguments, that $G^+_{[\gamma_1, \gamma_2]}
$ and $(2n)^{1/2} (T_{[\gamma_1, \gamma_2]}-1)^{1/2}$ have the same limiting distribution under the null hypothesis if $\phi(x) = \{x (1-x)\}^{-1}$.

Let $n,i \rightarrow \infty$ with $i/n = p$. Then if $p \in [\gamma_1,\gamma_2]$, the variable $u_{i}$ is asymptotically normal with mean $p$ and variance proportional to $1/n$ and therefore $u_{i} = p + O_p(n^{-1/2})$. The variable $\ell_i$, as a function of $u_i$, has first and second derivatives 0 and $p^{-1}(1-p)^{-1}$. Therefore, by Taylor expansion we find $\ell_i = 1 + (p-u_i)^2/[2 p (1-p)] + o_p(1/n)$ \citep[Prop. 6.1.5]{brockwell91}. Hence
\[(2n)^{1/2} [\ell_i-1]^{1/2}  = n^{1/2}|p-u_i|\left/\{p (1-p)\}^{1/2}\right. + o_p(1).\]
Let $S = \{i:u_i \leq i/n; \gamma_1 \leq u_i \leq \gamma_2\}$. Heuristically ignoring the influence of $o_p(1)$ terms (of which there are a growing number with $n$), assume that the limiting distribution of 
\begin{align*}
(2n)^{1/2}[T_{[\gamma_1, \gamma_2]}-1]^{1/2} &= \max\{(2n)^{1/2} [\ell_i-1]^{1/2}: i \in S\}
\end{align*}
is that of
\begin{equation*}
H_{[\gamma_1, \gamma_2]}=\max\left\{n^{1/2}(p-u_i)\left/\{p (1-p)\}^{1/2}\right.: i \in S\right\},
\end{equation*}
setting $H_{[\gamma_1, \gamma_2]}$ to zero when $S$ is empty. The absolute value was removed because $i \in S$ guarantees $u_i \leq p$. Replacing $p$ by $\hat{F}(y_i)$ and $u_i$ by $F(y_i)$  in the numerator of $H_{[\gamma_1, \gamma_2]}$, and replacing $p$ by $F(y_i)$ in the denominator (by the almost sure convergence of $F(y_i)$ to $p$), we find that $H_{[\gamma_1, \gamma_2]}$ is also the limiting random variable of $G^+_{[\gamma_1, \gamma_2]}$ if $\phi(x) = [x (1-x)]^{-1}$.

To prove this more rigorously we would need a better understanding of the joint behaviour of the $o_p(1)$ terms. The convergence of the two distributions is illustrated in Figure \ref{fig:asymp}, with $n=1000$ and $[\gamma_1, \gamma_2]=[.01, .99]$. By simulation we found that the fit seemed to deteriorate as $\gamma_1 \rightarrow 0, \gamma_2 \rightarrow 1$, and in the limit the result does not seem to hold. This is not necessarily surprising, for instance it is noted in \citet{chicheportiche12} that with $\gamma_1 = 1/(n+1)$ and $\gamma_2 = n/(n+1)$ the asymptotic distribution of a two-sided version of $G_{[\gamma_1, \gamma_2]}$ still depends on $n$, and in fact the asymptotic theory of weighted K-S tests generally relies on $\phi$ being bounded over the unit interval \citep[p.196]{anderson52}.

\section{Example: information flow in the Enron email corpus}
\label{sec:enron}
The Enron email corpus is a dataset that comprises emails sent and received by about 150 senior executives at the Enron Corporation, over the period 1998 to 2002. Although it is well-known to suffer from various integrity problems, it makes an attractive real data example because it is publicly available and many contemporary readers will be familiar with emailing behaviour. The dataset we analyse was downloaded from \verb+http://bailando.sims.berkeley.edu/enron_email.html+
 and reprocessed for our application. Only emails sent during the year 2001 were retained, because the record appears to be cleanest for that year. Some further effort was then needed to obtain reliable data. For example, many different email addresses can correspond to the same identity since an individual, John Smith say, could appear as any of john.smith, x..smith, jsmith $@$ either enron.com or ect.enron.com and more. 
Following \citet{perry13} we discarded emails sent to more than $5$ recipients, a subjectively chosen threshold that allows us to focus on inter-personal communications rather than company-wide announcements. 

Results will be presented for an individual, hereafter identified as $o$, who emailed frequently over the year, and for whom there are 12 individuals (of the 150 above) who contact $o$ and that $o$ contacts back. These are referred to by the identifiers $1,\dots,12$. 

Our example will seek to determine whether $i \rightarrow o$ triggers $o \rightarrow j$, denoted $i \rightarrow o \leadsto o \rightarrow j$ using only the timing of events. When $i=j$, a significant test is evidence of \emph{reciprocation} (or $o$ responding to emails), otherwise it suggests \emph{information flow}. The point processes generated by $i \rightarrow o$ and $o \rightarrow j$ replace $A$ and $B$ respectively in Section~\ref{sec:triggering}.

\begin{figure}[htp]
  \centering
  \includegraphics[width=8cm]{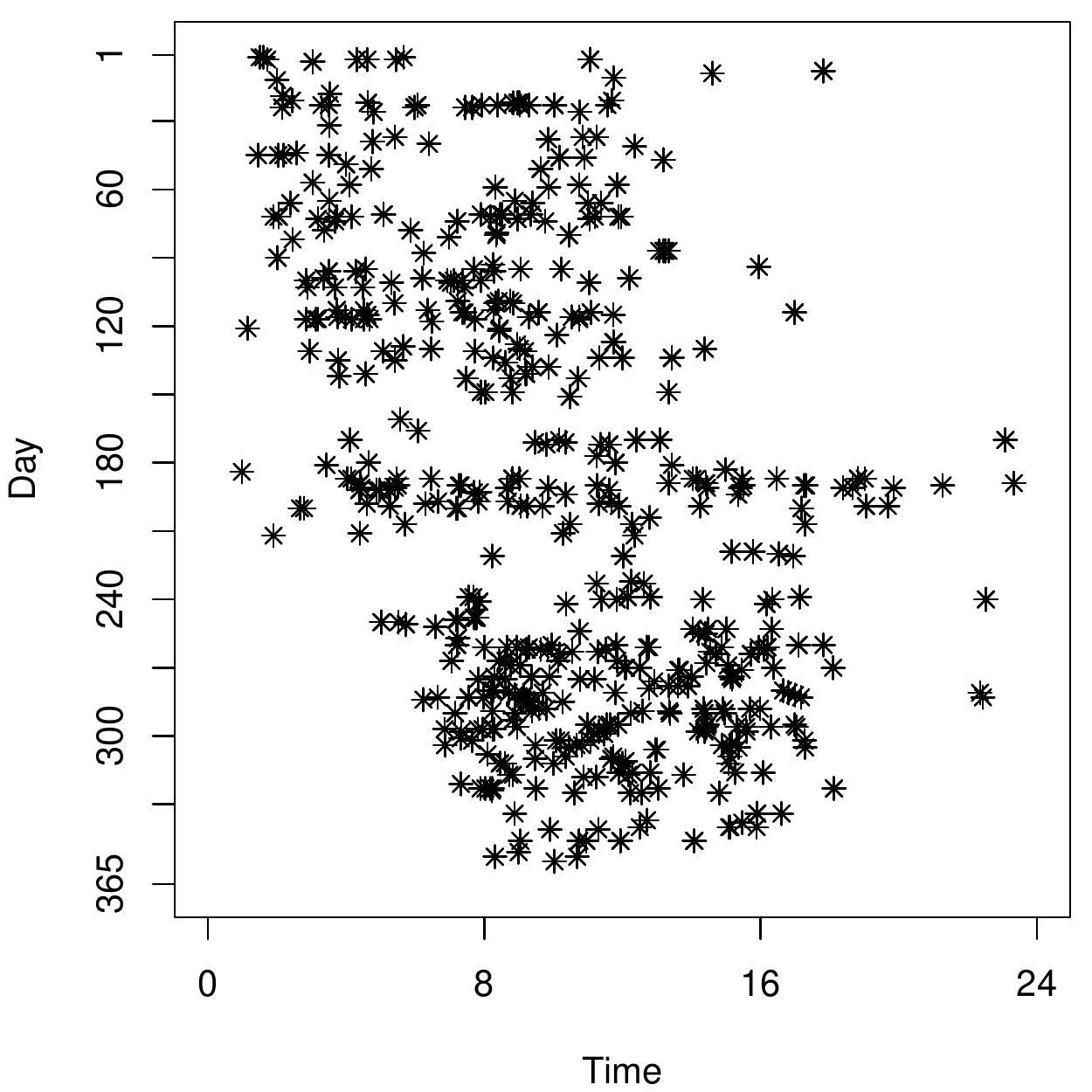}
\label{fig:pol}
\caption{Emailing behaviour of an individual in the Enron dataset. The stars show the event times, with the y-axis indicating the day and x-axis the time of day.}
\end{figure}
Figure~\ref{fig:pol} presents $o$'s sent email times, with the y-axis showing the day and x-axis the time of day of each event. This brings out a daily pattern in emailing behaviour; for example $o$ is markedly less active between the hours of 1600 and midnight (in some unknown time zone). On the other hand, looking vertically, there is some suggestion of behavioural change at a larger time-scale. For example there appears to be a busy period around the middle of year (specifically, July). 

A Bayesian model was implemented to estimate $o$'s intensity function over the pooled sample of all sent email times. This model attempts to capture the effects mentioned above, by jointly fitting multiple (wrapped) changepoints over the day and multiple changepoints over the year. Samples of the posterior intensity were computed, then standardized to integrate to one and then averaged. The resulting intensity is assumed to apply on every edge $o \rightarrow j, j = 1, \ldots, 12$, providing $r$ in \eqref{eq:model}. Details of the model and inference are given in the supplementary material. 

\begin{figure}[htp]
  \centering
  \includegraphics[width=8cm]{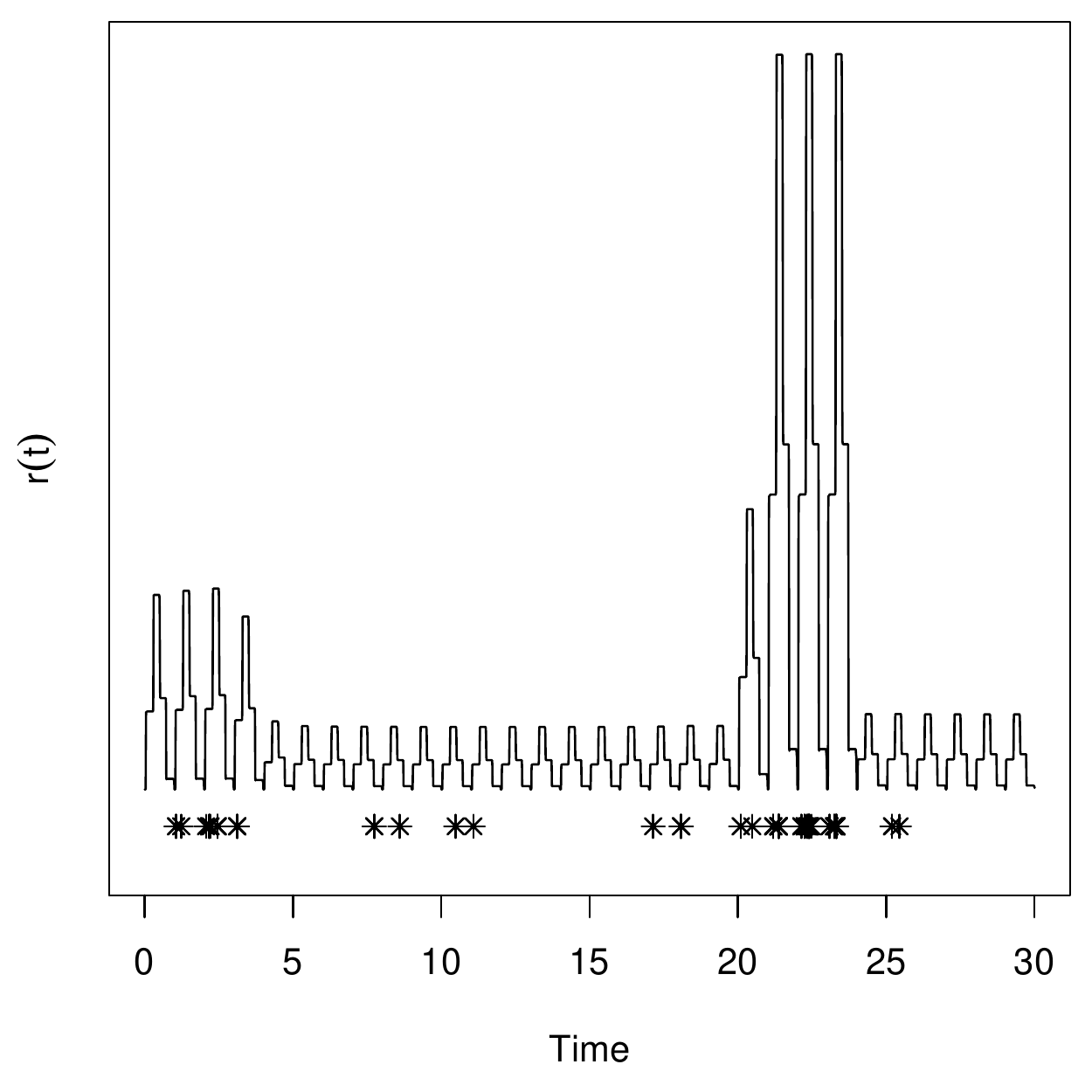}
  \caption{Emailing behaviour of an individual in the Enron dataset: fitted intensity }
  \label{fig:intensity_fit}
\end{figure}

Figure~\ref{fig:intensity_fit} illustrates our model fit to $o$'s emailing behaviour. The crosses denote event times, now on the $x$-axis. For visibility purposes only data from the first 30 days are presented. The line is $\hat r(t)$, which we will use in place of $r$ in \eqref{eq:model}. The model finds a unimodal daily pattern and, for instance, a period of high activity between the 20th and 25th of January.

\begin{figure}[htp]
  \centering
  \includegraphics[width=8cm]{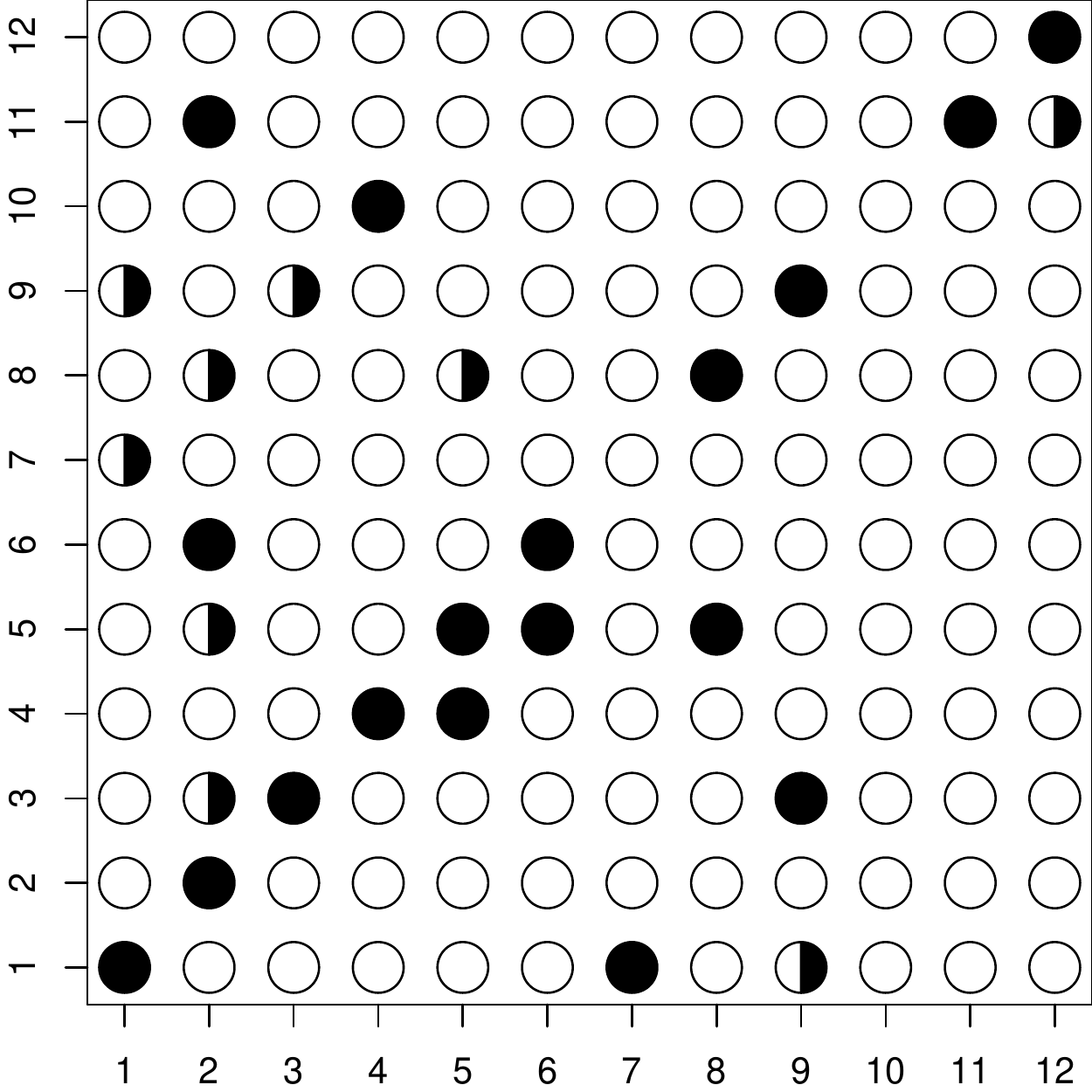}
  \caption{Evidence of $i \rightarrow o \leadsto o \rightarrow j$, for $i,j=1, \dots, 12$. The black circles indicate results retained by a false discovery rate set at 10\%, the half-circles $p$-values not retained by this criterion but still smaller than 0.05, while the white circles indicate a p-value greater than 0.05.}
\label{fig:JDft}
\end{figure}

In Figure~\ref{fig:JDft}, the results of testing $i \rightarrow o \leadsto o \rightarrow j$ are shown for $i,j=1, \dots, 12$, limiting the range of $\tau$ to $\tau_{\max}=\text{1 week}$, as discussed in Section~\ref{sec:time_limit}. The black circles are the p-values that are retained for analysis using a false discovery rate of $10\%$ \citep{benjamini95}. The half-circles are p-values less than $5\%$ and the white are the (not significant at the $5\%$ level) p-values that remain.

Most of the entries on the bottom-left to top-right diagonal are black, meaning that there is compelling evidence for reciprocation. Because reciprocation is largely to be expected, the two white circles on that diagonal warrant additional inspection: they indicate a lack of evidence for $o$ responding to emails from identities $7$ or $10$.
\begin{description}
\item[$7 \rightarrow o \not\leadsto o \rightarrow 7$:] There is only one email from $7$ to $o$ and one other from $o$ to $7$. They are sent about one month apart (and appear to be unrelated judging by their subject-lines). The p-value is automatically 1 because $\tau_{\max}$ was set to a week.
\item[$10 \rightarrow o \not\leadsto o \rightarrow 10$:] This example is more interesting. The $p$-value is only 0.28 despite there being 14 emails from $10$ to $o$ and 9 from $o$ to $10$, the most coincidental email times falling in July, about $3.5$ hours from each other. The reason why no effect is detected is in part because $o$ is estimated to be relatively busy in July, with $\rho(\text{July})\approx 0.13$ as opposed to the average $1/12 \approx 0.08$, meaning that we are less sensitive to coincidental timings during that month than at other times. In fact, upon inspecting the subject-lines of $10 \rightarrow o$ and $o \rightarrow 10$, it does appear as if $10$ and $o$ do not reciprocate. For example, the subject-lines of the two most coincidental emails are ``FW: Enron Complaint'' and ``Dunn hearing link?'', which are not obviously related. 
\end{description}

Consider now the p-values that were retained with a false discovery rate controlled at $10\%$. Table~\ref{tab:subject_lines} shows the subject-lines of the emails that the test based its decision on. More precisely, for each of the retained $i,j$ pairs, we find the closest two $i \rightarrow o$ and $o \rightarrow j$ events, subject to the former preceding the latter. Thus we have the `most triggering' email event $e$, and its subject-line is displayed first. Next, we display the subject-lines of all the emails from $o$ to $j$ that fall within $[e, e+\hat{\tau}]$, for the $\hat{\tau}$ used by the test. 

Table~\ref{tab:subject_lines} shows that the method succeeds in picking out `real' excitation periods. Consider for example the most significant detection $11 \rightarrow o \leadsto o \rightarrow 11$. Four emails fall within $\hat{\tau}$ of the `most triggering' email. These all have the subject-line ``RE: DWR - Gas Daily'' whereas the subject-line of the original email by identity 11 is ``DWR - Gas Daily''. Furthermore, the next email in the $o \rightarrow 11$ sequence, i.e., the first that is estimated not to be triggered, has a different subject-line ``RE: DWR and Edison Meetings''.

Although the method largely found evidence of reciprocation, there are some places where real information flow was identified. The email from identity $9$ to $o$ with subject ``California Update--Legislative Push Underway'' being followed by a string of emails from $o$ to $3$ with subject ``Re: California Update--Legislative Push Underway'' is a particularly compelling example. A concern could be that identity $3$ was simply `cc'ed' while  $o$ was responding to identity $9$. This is not the case: two of the four emails displayed, the second and the fourth, are sent from $o$ directly to $3$ with no other party involved. For reference, there were 28 emails in $[e, e+\hat{\tau}]$. Only those with the matching subject-lines are displayed. The vertical dots indicate the position of those omitted.

Of course, there are also a number of false positives in the results, notably the detection $5 \rightarrow o \leadsto o \rightarrow 4$. These are at least in part due to our use of an overly simple null model, which in particular fails to capture local bursts of activity prevalent in emailing behaviour.

\renewcommand{\arraystretch}{.8}
\begin{table}
  \caption{\label{tab:subject_lines}Subject-lines of sent emails that are estimated to be triggered. Further details in main text.}
\fbox{\scriptsize%
\begin{tabular}{c c r@{:}r@{:}l p{9cm}}
Pattern & P-value & \multicolumn{3}{c}{Time lag} & Subject\\ \hline
$ 11 \rightarrow o \leadsto o \rightarrow 11 $ & $ 9.6\times 10^{-12} $ &0&00&00&{\bf  DWR - Gas Daily }\\
 & & 0 & 00 & 58 & RE: DWR - Gas Daily\\
 & & 0 & 58 & 31 & RE: DWR - Gas Daily\\
 & & 1 & 06 & 49 & RE: DWR - Gas Daily\\
 & & 1 & 27 & 12 & RE: DWR - Gas Daily\\
$ 2 \rightarrow o \leadsto o \rightarrow 2 $ & $ 4.7\times 10^{-9} $ &0&00&00&{\bf  RE: CPUC Questions on DA }\\
 & & 0 & 00 & 19 & RE: CPUC Questions on DA\\
$ 4 \rightarrow o \leadsto o \rightarrow 4 $ & $ 4.9\times 10^{-8} $ &0&00&00&{\bf  RE: Transwestern Hearing }\\
 & & 0 & 14 & 00 & RE: Transwestern Hearing\\
 & & 5 & 01 & 00 & RE: Transwestern Hearing\\
$ 12 \rightarrow o \leadsto o \rightarrow 12 $ & $ 1.5\times 10^{-6} $ &0&00&00&{\bf  RE: CA Unbundling }\\
 & & 0 & 04 & 58 & RE: CA Unbundling\\
$ 9 \rightarrow o \leadsto o \rightarrow 3 $ & $ 4.6\times 10^{-5} $ &0&00&00&{\bf  California Update--Legislative Push Underway }\\
 & & 0 & 51 & 00 & Re: California Update--Legislative Push Underway\\
 & & 1 & 03 & 00 & Re: California Update--Legislative Push Underway\\[-.2cm]
 & &  \multicolumn{3}{c}{} & \vdots\\[-.4cm]
 & & 10 & 51 & 00 & Re: California Update--Legislative Push Underway\\
 & & 11 & 03 & 00 & Re: California Update--Legislative Push Underway\\[-.2cm]
 & &  \multicolumn{3}{c}{}  & \vdots\\[-.4cm]
$ 1 \rightarrow o \leadsto o \rightarrow 1 $ & $ 4.7\times 10^{-5} $ &0&00&00&{\bf  Re: Comments to Gov\'s Proposals }\\
 & & 0 & 02 & 00 & Re: Comments to Gov\'s Proposals\\
 & & 5 & 39 & 00 & RE: Additional Materials\\
 & & 21 & 37 & 00 & Update from EES Call this Morning\\
$ 3 \rightarrow o \leadsto o \rightarrow 3 $ & $ 9.3\times 10^{-5} $ &0&00&00&{\bf  Re: Pescetti }\\
 & & 0 & 03 & 00 & RE: Pescetti\\
$ 9 \rightarrow o \leadsto o \rightarrow 9 $ & $ 3.7\times 10^{-4} $ &0&00&00&{\bf  California Update--Legislative Push Underway }\\
 & & 0 & 51 & 00 & Re: California Update--Legislative Push Underway\\
 & & 10 & 51 & 00 & Re: California Update--Legislative Push Underway\\
$ 2 \rightarrow o \leadsto o \rightarrow 6 $ & $ 8.6\times 10^{-4} $ &0&00&00&{\bf  HERE IS MY DRAFT }\\
 & & 0 & 09 & 00 & Re: FW: SoCalGas Capacity Forum\\
$ 6 \rightarrow o \leadsto o \rightarrow 6 $ & $ 1.7\times 10^{-3} $ &0&00&00&{\bf  Re: FW: SoCalGas Capacity Forum }\\
 & & 2 & 22 & 00 & Re: FW: SoCalGas Capacity Forum\\
$ 5 \rightarrow o \leadsto o \rightarrow 5 $ & $ 2.1\times 10^{-3} $ &0&00&00&{\bf  Re: Response to ORA/TURN petition }\\
 & & 0 & 03 & 00 & Re: Response to ORA/TURN petition\\
$ 8 \rightarrow o \leadsto o \rightarrow 5 $ & $ 2.2\times 10^{-3} $ &0&00&00&{\bf  RE: Call to Discuss Possible Options to Mitigate Ef\dots }\\
 & & 1 & 18 & 29 & Re:\\
$ 4 \rightarrow o \leadsto o \rightarrow 10 $ & $ 2.4\times 10^{-3} $ &0&00&00&{\bf  RE: Transwestern Hearing }\\
 & & 1 & 31 & 13 & Attorneys\\
$ 8 \rightarrow o \leadsto o \rightarrow 8 $ & $ 3.8\times 10^{-3} $ &0&00&00&{\bf  RE: Call to Discuss Possible Options to Mitigate Ef\dots }\\
 & & 2 & 11 & 29 & RE: Call to Discuss Possible Options to Mitigate Ef\dots\\
$ 5 \rightarrow o \leadsto o \rightarrow 4 $ & $ 7.1\times 10^{-3} $ &0&00&00&{\bf  FW:  EPSA report }\\
 & & 41 & 33 & 00 & RE: Transwestern Hearing\\
 & & 46 & 20 & 00 & RE: Transwestern Hearing\\
$ 2 \rightarrow o \leadsto o \rightarrow 11 $ & $ 7.5\times 10^{-3} $ &0&00&00&{\bf  Willie Brown INFO }\\
 & & 0 & 03 & 48 & RE: Socal Storage Projects\\
$ 7 \rightarrow o \leadsto o \rightarrow 1 $ & $ 8.4\times 10^{-3} $ &0&00&00&{\bf  Governor Davis\' Press conference Highlights -- wil\dots }\\
 & & 4 & 23 & 00 & \\
 & & 5 & 06 & 00 & Email for Transmittal from Ken Lay to Senator Brult\dots\\
$ 6 \rightarrow o \leadsto o \rightarrow 5 $ & $ 9.1\times 10^{-3} $ &0&00&00&{\bf  Re: FW: SoCalGas Capacity Forum }\\
 & & 2 & 22 & 00 & Re: FW: SoCalGas Capacity Forum\\
\end{tabular}}
\end{table}

\section{Discussion}
The case where the null intensity of the tested process is unknown has only been treated heuristically, by replacing it with a posterior expectation. A more formal treatment might proceed as follows. Suppose that under the null hypothesis $r$ has a model with unknown parameters $\theta$ for which we can calculate a confidence set with coverage probability $1-\epsilon$, for some specified (small) $ \epsilon \in [0,1]$. This induces a confidence set $S$ on $r$ with the same coverage. Let $T^* = \sup\{T(s):s \in S\}$, where $T(s)$ is the test statistic that would be computed in Algorithm \ref{alg:lik} if $r$ was known to be $s$. Then by Bonferroni correction $\epsilon+\{1-F_n(T^*)\}$ provides a conservative p-value in the case where $r$ is unknown, see e.g. \citet{silvapulle1996test}. It would be interesting to investigate the choice of $\epsilon$ and the computation of $T^*$ for some generic models for $r$.

\section*{Supplementary material}
Supplementary material available at the authors' website includes a proof of the validity of Algorithm \ref{alg:lik}, an illustration of the difficulty of using a conditional intensity in place of $r$, a number of further extensions of the test to detect different forms of dependence, and further details on the Bayesian model and inference used in Section ~\ref{sec:enron}.
\appendix
\section*{Appendix}
\begin{proof}[Proof of Lemma~\ref{lem:uniform}]
 Notice that $K(y) =  \#\{b_i \in \Tr(y)\}$ defines a point process over  $[0,\max\{a_{i+1}-a_i: i = 1, \dots, m\})$, temporarily defining $a_{m+1} = L$. Its events are obtained from the event times of $B$ by superposing the segments $[a_i, a_{i+1})$, aligning to the left.

The intensity of $K(y)$ is therefore $h(y) =\sum_{\{t: t-a(t)=y\}} \lambda_{B}(t)$ and its compensator is
\[H(y) = \int_0^y h(s) ds = \begin{cases} \lambda_1\rho\{\Tr(y)\} & y \leq \tau, \\
\lambda_1\rho\{\Tr(\tau)\}+\lambda_2[\rho\{\Tr(y)\} - \rho\{\Tr(\tau)\}] & y > \tau. \end{cases}\]
If $F$ is a non-decreasing function on a sub-interval $X$ of $\R$ with image $Y$, we define $F^{-1}(y) = \inf\{x \in X: F(x) \geq y\}$ for $y \in Y$. Because $\rho$ is continuous, $\rho\{\Tr(y)\}$ is a continuous non-decreasing function $\mu$, say, of $y$. Then $\mu^{-1}$ is right-continuous and non-decreasing with jumps at an at most countable set of values of $x$ corresponding to intervals where $\mu$ is constant \citep[p.420]{daley07}. Now let $U(x) = K(\mu^{-1}(x))$ for $x \in [0,1)$. The $q$th event time of $U$ is 
\begin{align*}
\inf(x: U(x) \geq q)& = \inf(x: K\left(\mu^{-1}(x)\right) \geq q)\\
&=\inf\left\{x: \mu^{-1}(x) \geq k_q\right\}\\
& = \rho\left\{\Tr\left(k_q\right)\right\}\\
& = u_q,
\end{align*}
where $k_q$ is $q$th smallest response time, i.e., the $q$th event time of $K$. Using the well-known time change theorem \citep[p.421]{daley07},  $K(H^{-1}(z))$ is a Poisson process with rate 1 for $z \in [0, H(L))$. Since $U(x) = K(H^{-1}(\lambda_1 x))$ for $ x \in [0,\rho\{\Tr(\tau)\}]$, $U(x)$ is homogeneous Poisson with intensity $\lambda_1$ over that range. By a similar argument $U(x)$ has intensity $\lambda_2$ over $(\rho\{\Tr(\tau)\},1)$.
\end{proof}

\begin{proof}[Proof of Corollary~\ref{cor:consistence}]
Let $w=\rho\{\Tr(\tau)\}$. The conditions on $\tau$ guarantee that $0 < w < 1$, whereas $\lambda_1,\lambda_2 \rightarrow \infty$ guarantee $n \rightarrow \infty$. Conditional on $n$, the variables $u_1,\ldots,u_n$ in disorder are independent and identically distributed with density 
\begin{equation*}
d(x) = \begin{cases} \alpha & x \leq w, \\
\beta & x>w, \end{cases}
\end{equation*}
where $\alpha w + \beta (1-w)=1$ and $\alpha = c \beta$. \citet{chernoff1956estimation} showed that the (unconstrained) maximum likelihood estimator $\hat w$ for $w$ is consistent. Let $\hat l$ denote the number of $u_i \leq \hat w$. The maximum likelihood estimator for $\alpha$ is $\hat\alpha^* = \hat w \hat l /n$. This is consistent because $\hat \alpha^* = \hat w l/n \pm \hat w e /n$ where $l$ is the true number of $u_i \leq w$ and $e$ is the number of misclassifications. We have $e/n = o_p(1)$ by the consistency of $\hat w$ and $l/n = \alpha w + o_p(1)$, by the consistency of the maximum likelihood estimate of the Bernoulli parameter. By a similar argument the unconstrained maximum likelihood estimate for $\beta$ is also consistent. Therefore the constrained maximum likelihood estimates for $w,\alpha,\beta$ such that $\hat \alpha = \hat w \hat l/n \geq (1-\hat w) (n-\hat l)/n=\hat \beta$ are consistent, since the probability that the constrained and unconstrained versions disagree tends to zero. We then verify that $\hat w = \rho\{\Tr(\hat\tau)\}$, for the estimate of $\tau$ proposed in Algorithm~\ref{alg:lik}.

The condition on $r$ ensures that $\hat \tau$ is consistent when $\hat w$ is consistent, proving the first part of the claim. We also have
\begin{align*}
T&=\hat \alpha^{K(\hat{\tau})/n} \hat \beta^{(n-K(\hat\tau)/n}\\
&= \{\alpha+o_p(1)\}^{K(\tau)/n+o_p(1)} \{\beta+o_p(1)\}^{(n-K(\tau))/n+o_p(1)} \\
&= \alpha^{K(\tau)/n}\beta^{(n-K(\tau))/n}+o_p(1).
\end{align*}
Under $H_0$ the likelihood of $B$ conditional on $n$, $\ell_0(B\mid n)$, is constant. The non-vanishing part of the above is a monotonic function of the likelihood of $B$ conditional on $n$ under $H_1$, $\ell(B\mid n; \tau,\lambda_1,\lambda_2)$. Thus $T$ becomes a monotonic function of  $\ell(B\mid n; \tau,\lambda_1,\lambda_2)/\ell_0(B\mid n)$.
\end{proof}

\bibliographystyle{apalike} 
\bibliography{ppdependence}

\end{document}